\documentclass[aps,10pt,reprint,groupedaddress,superscriptaddress, prx]{revtex4-2}

\usepackage{amssymb}
\usepackage{graphicx}
\usepackage{amsmath}
\usepackage{multirow}

\usepackage[bookmarks=true,colorlinks,citecolor=blue,urlcolor=blue]{hyperref}
\usepackage{braket}
\usepackage{babel}
\usepackage{blindtext}
\usepackage[compat=1.1.0]{tikz-feynman}
\usepackage[normalem]{ulem}
\usepackage{physics}

\usepackage{mathtools}

\DeclarePairedDelimiter\ceil{\lceil}{\rceil}

\definecolor{mypurp}{rgb}{0.35, 0, 0.7}

\usepackage{amsthm}
\usepackage{txfonts}
\newtheorem{prop}{Proposition}
\newtheorem{theorem}{Theorem}
\newtheorem{lemma}[theorem]{Lemma}

\theoremstyle{definition}

\begin{document}

% affiliation
\newcommand{\NottinghamA}{\affiliation{School of Physics and Astronomy, University of Nottingham, Nottingham, NG7 2RD, UK}}
\newcommand{\NottinghamB}{\affiliation{Centre for the Mathematics and Theoretical Physics of Quantum Non-Equilibrium Systems, University of Nottingham, Nottingham, NG7 2RD, UK}}
\newcommand{\TUM}{\affiliation{Technical University of Munich, TUM School of Natural Sciences, Physics Department, 85748 Garching, Germany}}
\newcommand{\MCQST}{\affiliation{Munich Center for Quantum Science and Technology (MCQST), Schellingstr. 4, 80799 M{\"u}nchen, Germany}}

\def\papertitle{{Model-Independent Learning of Quantum Phases of Matter \\with Quantum Convolutional Neural Networks}}
\title{\papertitle}

\author{Yu-Jie Liu} \TUM \MCQST
\author{Adam Smith} \NottinghamA \NottinghamB
\author{Michael Knap}  \TUM \MCQST
\author{Frank Pollmann}  \TUM \MCQST

%\date{\today}

\begin{abstract}

Quantum convolutional neural networks (QCNNs) have been introduced as classifiers for gapped quantum phases of matter. Here, we propose a model-independent protocol for training QCNNs to discover order parameters that are unchanged under phase-preserving perturbations. We initiate the training sequence with the fixed-point wavefunctions of the quantum phase and add translation-invariant noise that respects the symmetries of the system to mask the fixed-point structure on short length scales. 
We illustrate this approach by training the QCNN on phases protected by time-reversal symmetry in one dimension, and test it on several time-reversal symmetric models exhibiting trivial, symmetry-breaking, and symmetry-protected topological order. The QCNN discovers a set of order parameters that identifies all three phases and accurately predicts the location of the phase boundary. The proposed protocol paves the way towards hardware-efficient training of quantum phase classifiers on a programmable quantum processor.
\end{abstract}

\maketitle

Phases of matter are traditionally identified by measuring order parameters, including local order parameters for symmetry-breaking phases and string order parameters for one-dimensional symmetry-protected topological (SPT) phases~\cite{Nijs:1989,Kennedy:1992,perez:2008,Pollmann:2010,schuch:2011,chen:2011}.
Finding a suitable string order parameter can be difficult in general, in particular, without the presence of additional symmetries~\cite{Ppollmann:2012}.
Nonetheless, non-local order parameters, utilizing multiple copies of the system, can directly extract the topological invariant of one-dimensional (1D) SPT phases for global symmetries~\cite{haegeman:2012, Ppollmann:2012}. They can also be probed by randomized measurements~\cite{elben:2020} at a cost that scales exponentially with the subsystem size of interest.

\begin{figure}[t!]
    \centering
    \includegraphics{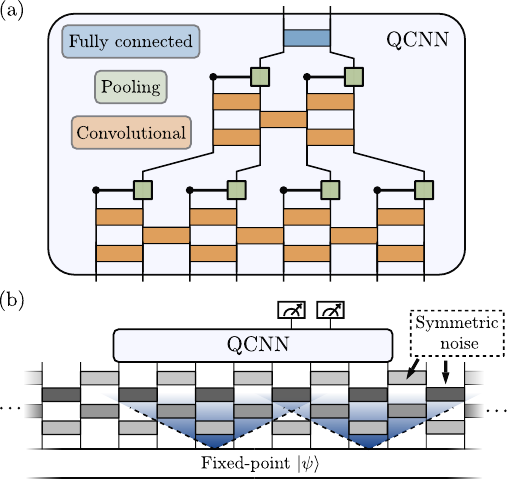}
    \caption{Architecture of an 8-qubit QCNN and training data. (a) The QCNN consists of three building blocks~\cite{cong:2019}. Convolutional layers consist of two-qubit gates (orange). The pooling is achieved by a set of controlled rotations (green). The gates in the pooling and the convolutional layers can be either translation-invariant or independently parametrized. The fully connected layer in this case is a two-qubit gate (blue). (b) Random local symmetric two-qubit gates create finite correlations that mask the microscopic details of the state. In the shown example, four layers of symmetric noise ($L_{\text{noise}}=4$) are applied to a fixed-point in the thermodynamic limit. Each layer is translationally invariant and independently sampled.}
    \label{fig:qcnn}
\end{figure}

Recently, classical and quantum machine learning approaches have been introduced to tackle the task of quantum phase classification; see e.g. Refs.~\cite{Broecker:2017,huembeli:2018,dong:2019,Bohrdt:2019, carleo:2019, cong:2019, huang:2022,monaco:2023,sadoune:2023}. Quantum circuit classifiers, such as quantum convolutional neural networks (QCNNs)~\cite{cong:2019}, as demonstrated experimentally in Ref.~\cite{Herrmann:2022}, naturally provide a quantum machine learning architecture to learn observables for the classification of phases. One advantage of these classifiers is, that the learned observable can be efficiently measured on a quantum device. 
Thus, an interesting question is to study whether a quantum machine learning approach, such as the training of a QCNN, can automate the discovery of (non-local) order parameters that characterize the phases and are efficient to measure experimentally. 
The training of a quantum classifier faces challenges. For example, a large amount of labelled training data, in the form of ground states, is needed. These states could be obtained by adiabatic~\cite{farhi:2000,ge:2019} or variational algorithms~\cite{Peruzzo:2014,McClean:2016}. 
However, the generation of a large amount of ground states can become infeasible, in particular, for noisy intermediate scale quantum devices. Another prominent obstacle is vulnerability or over-fitting of the QCNN, caused by training on a specific class of ground states.

%Another prominent obstacle is that training on a specific class of models can lead to a highly vulnerable or over-fitted QCNN, as the QCNN could pick up some convenient properties that are only specific to the class of models provided in the training. 

To address the aforementioned issues, we propose a model-independent quantum protocol for training QCNNs using minimal information about the gapped quantum phases, which includes the fixed-point wavefunctions and the symmetry group of the system~\cite{gu:2010,zeng:2015}. 
We train the QCNN with synthetic training data by first constructing the fixed-point wavefunctions, which are typically efficiently prepared~\cite{verstraete:2005,gu:2009a}, and then apply a finite number of layers of random symmetric local gates. Each layer is translationally invariant and independently sampled. 
Randomness helps prevent the QCNN from learning local (non-universal) properties of the states by masking the local structure of the fixed-point.

\textit{\textbf{Quantum phase classification.---}}The architecture of QCNNs consists of convolutional, pooling, and fully-connected layers~\cite{cong:2019}, as shown in Fig.~\ref{fig:qcnn}a. Each convolutional layer is a finite-depth circuit of local unitary gates. The pooling layer is a set of parallel controlled single-qubit rotations, where the control qubits are discarded.
%after which the control qubits remain untouched and the target qubits enter the next level. 
During the training, the controlled rotations can be absorbed into the convolutional layer. Originally, the unitary gates on the same circuit layer were chosen to be identical everywhere~\cite{cong:2019}. Besides this uniform ansatz, we also investigate independently parametrized gates. Despite having an extensive number of free parameters, and hence an increase in the difficulty of training, both the uniform and the generalized QCNN are barren-plateau free~\cite{pesah:2021}. Before the readout, a fully connected layer ``summarizes'' the information into the measured qubits. This layer is a multi-qubit gate that acts on all the qubits at the final level of the network.

Here, we focus on 1D systems and consider the following task: Suppose we have $M$ phases protected by a symmetry group $G$ and a set $\Psi_G = \{\ket{\psi_1},\cdots,\ket{\psi_{M}}\}$, where each $\ket{\psi_m}\in \Psi_G$ is a fixed-point wavefunction for each phase under the real-space renormalization group flow~\cite{verstraete:2005, vidal:2007,levin:2007,gu:2009a,evenbly:2015}. We aim to find a QCNN that predicts the phase of any input ground state $\ket{\psi}$ of a symmetric Hamiltonian.

We tackle the classification task by training an $N$-qubit QCNN that acts on an infinite system (Fig.~\ref{fig:qcnn}b). 
We choose the number of readout qubits at the fully connected layer to be $\ceil*{\log_2M}$ and associate each phase with a bitstring label $s\in\{1,2,\cdots,M\}$. The probability for each bitstring $\ket{s}$ is interpreted as the QCNN's confidence score for that phase, and the prediction is the phase with the largest probability.
%probability of the phase predicted by the QCNN. The prediction is then made by selecting the phase with the largest probability. 
A quantum phase transition is marked by the change of the label with the highest probability.
This contrasts the original QCNN considered in Ref.~\cite{cong:2019}, which produces an order parameter that vanishes for one phase and is non-zero for the other phases.

\textit{\textbf{Model-independent training.---}}To train an $N$-qubit QCNN, we use the stability of the quantum phases under finite-depth symmetric quantum circuits. Two ground states $\ket{\psi_1}, \ket{\psi_2}$ belong to the same phase if and only if they are related by a finite-depth local quantum circuit 
$
    \ket{\psi_1}\sim\prod_k \hat U^{(k)}\ket{\psi_2},
$
where $\hat U = \prod_k \hat U^{(k)}$ is a product of layers of local unitaries that can be continuously connected to the identity~\cite{gu:2010,zeng:2015}. 
When the system has certain symmetries, $\hat U$ needs to be symmetric as well.

We generate the training data with the following steps (sketched in Fig.~\ref{fig:qcnn}b):
\begin{enumerate}
    \item Randomly pick a label $m\in \{1,\cdots,M\}$ and prepare the fixed-point wavefunction $\ket{\psi_m}\in\Psi_G$.
    \item Apply $L_{\text{noise}}$ layers of random symmetric local two-qubit gates, for $L_{\text{noise}}<N/2$. 
\end{enumerate}
The requirement of $L_{\text{noise}}<N/2$ comes from the finite size of the QCNN. Namely, if the correlation length created by the noise becomes comparable to the QCNN's size, the phases are no longer distinguishable by the QCNN. 
In practice, we first train the QCNN with a single layer of noise then increase the number of layers one-by-one as we achieve convergence. We continue until test accuracy falls below a threshold.
%In practice, we first train the QCNN by a single layer of noise. Once convergence is achieved, we increase the number of noise layers by one and again optimize the QCNN. This procedure is continued, until the test accuracy falls below a threshold.
%
We restrict ourselves to the simplest case of two-qubit gates for the noise. However, the scheme can be easily generalized to symmetric gates that act on more qubits.
%
% There is a simple physical picture behind the protocol. The fixed-point wavefunction $\ket{\psi_m}$ encodes the long-distance physics that characterizes the phase. The random symmetric layers effectively mask the short distance details within $\ket{\psi_m}$, and thus force the QCNN to learn the symmetry-invariant long-range order of $\ket{\psi_m}$.

\textit{\textbf{Training time-reversal symmetric phases in 1D.---}}To investigate the effectiveness of the protocol, we consider classification of gapped phases of 1D translationally invariant spin-$1/2$ chains with time-reversal symmetry generated by $T = \left(\prod_i X_i\right)K$, where $X$ is the Pauli-$X$ matrix and $K$ is complex conjugation. We focus on the cases where the translation symmetry is not spontaneously broken. The system hosts three phases~\cite{schuch:2011,chen:2011}: (i) The symmetry-breaking (SB) phase, where $T$ is spontaneously broken and the system has degenerate ground states; (ii) the trivial phase; and (iii) the SPT phase. In the thermodynamic limit, the ground state is unique for (ii) and (iii), but $T$ is fractionalized trivially and non-trivially, respectively. A fixed-point wavefunction can be found for each phase
\begin{align}
     \ket{\psi_{\text{SB}}} &= \frac{1}{\sqrt{2}}\left(\ket{\cdots 000\cdots}+\ket{\cdots 111\cdots}\right),\nonumber\\
      \ket{\psi_{\text{Trivial}}} &= \ket{\cdots +++\cdots},\qquad \ket{\psi_{\text{SPT}}} = \ket{CS},\label{eq:FP}
\end{align}
where the basis $Z\ket{0} = \ket{0}, Z\ket{1} = -\ket{1}$ is used, and $\ket{+} = (\ket{0}+\ket{1})/\sqrt{2}$. The state $\ket{CS}$ is the cluster state satisfying $Z_{i-1}X_iZ_{i+1}\ket{CS} = -\ket{CS},\ \forall i$. The set of time-reversal symmetric two-qubit unitary gates that continuously connects to the identity forms a Lie group $Q$ generated by the set of Pauli strings $P = \{iZ_1, iZ_{2},iZ_1Y_{2},iY_1Z_{2}, iZ_1X_{2},iX_1Z_{2}  \}$.

In the Supplemental Material (SM)~\cite{SM}, we prove that the QCNN that we aim to find does not exist without imposing translational invariance (TI), or other additional symmetries, on the input data. However, when we impose TI, there exists a set of observables that can be used to perfectly identify the phases in the thermodynamic limit~\cite{SM}.
Therefore, a QCNN for TI input states may be found and the prediction of the phase is obtained by applying a low-depth quantum circuit followed by local bitstring measurements.

\begin{figure}
    \centering
    \includegraphics{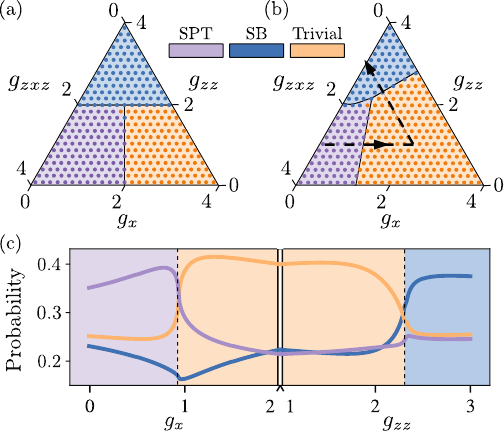}
    \caption{(a) QCNN prediction of the phase diagram of the cluster-Ising model $H_{CI}$ (filled circles); see Eq.~\eqref{eq:Ising_cluster}. The background and the boundaries are the theoretical phase diagram solved analytically. (b) QCNN prediction (filled circles) obtained for a perturbed cluster-Ising model $H_{pCI}$. The theoretical phase diagram is obtained by an iDMRG algorithm. (c) A zoom-in along the cut of the perturbed phase diagram in (b) indicated by the black dashed line. The QCNN prediction corresponds to the solid lines. The dashed vertical lines mark the location of the critical points predicted by iDMRG.}
    \label{fig:prediction}
\end{figure}

In principle we allow TI for arbitrary size unit cells. Here, we impose TI with a two-site unit cell by using a single two-qubit gate per layer, which is repeated across the system, see Fig.~\ref{fig:qcnn}b. 
The two-quibt gate is parametrized as $\exp(\sum_k \theta_k P_k)$, where $P_k \in P$, the set of generators for the symmetric noise, and each $\theta_k$ is randomly sampled from a uniform distribution.
In the following, we will focus on the QCNN in which every gate is independently parametrized.
To classify the three phases, the prediction is made by measuring two qubits at the end of the QCNN circuit. We assign labels to each bitstring output: $00\to\text{Trivial}$, $01\to\text{SB}$ and $10\to\text{SPT}$. The bitstring $11$ corresponds to an unsuccessful classification.

We provide details of the training in the SM~\cite{SM} and summarize the results here. The trained QCNN is tested on $10^5$ synthetic data constructed with different $L_{\text{noise}}$ to obtain a final test accuracy. We first train a $4$-qubit QCNN. The test accuracy reached by the 4-qubit QCNN is $87.21\%$ on the test data generated with $L_{\text{noise}} = 1$. This means that the $4$-qubit QCNN is not able to perfectly distinguish the phases of a 4-qubit subsystem when the correlation length of the system is roughly two sites. For an 8-qubit QCNN, the performance is drastically improved and attains a test accuracy of $100\%$ on engineered data with $L_{\text{noise}} = 1$ and of $97.37\%$ on data with $L_{\text{noise}} = 2$. 
% The high accuracy indicates that the 8-qubit QCNN finds a way to characterize the time-reversal symmetric phases (with short correlations).

In the SM~\cite{SM}, we also examine the ansatz with uniformly parametrized gates within the 8-qubit QCNN, which achieves similar performance to the non-uniform ansatz when increasing the depth of the convolutional layer from 3 to 5 in Fig.~\ref{fig:qcnn}a. Despite the deeper circuit, the uniform ansatz has fewer free parameters.

\textit{\textbf{Predictions on physical models.---}}
We now test the trained QCNN on different time-reversal symmetric physical models. We will use the 8-qubit QCNN and present 4-qubit QCNN data in the SM~\cite{SM}.

We first consider a cluster-Ising model~\cite{smacchia:2011,verresen:2017,smith:2022} where the phases are protected by time-reversal symmetry $T=\left(\prod_iX_i\right)K$. The Hamiltonian of the system is
\begin{equation}
    H_{CI} = g_{zxz}\sum_i Z_{i-1}X_iZ_{i+1} - g_{zz}\sum_i Z_iZ_{i+1} - g_x\sum_i X_i, \label{eq:Ising_cluster}
\end{equation}
where $g_{zxz}, g_{zz},g_{x}\geq 0$. Depending on the couplings, the symmetry protects three distinct phases---the trivial, the SB and the SPT phase.

We test the QCNN over the phase diagram of $H_{CI}$ in Eq.~\eqref{eq:Ising_cluster}. We do this by first finding the ground states of $H_{CI}$ using an infinite density matrix renormalization group algorithm (iDMRG)~\cite{dmrg:1993,vidal:2007a}, which are then input to the QCNN for classification. The results are shown in Fig.~\ref{fig:prediction}a. The background color marks the theoretical prediction, while the colored circles show the QCNN prediction. The theoretical phase diagram is obtained by mapping $H_{CI}$ to a free-fermion chain~\cite{verresen:2017}. We see that the QCNN accurately predicts the phase diagram. To test robustness, we add a perturbation $H_{pCI}= H_{CI}-g_x\sum_i X_{i}X_{i+1}$ that breaks the free-fermion mapping of the chain. The prediction of the same QCNN is shown in Fig.~\ref{fig:prediction}b. In this case, the theoretical phase diagram is obtained by a transfer-matrix approach ~\cite{Ppollmann:2012} based on iDMRG. The trivial phase is expanded in parameter space due to the additional coupling. The trained QCNN again accurately predicts the shifted phase boundary.
To take a closer look, in Fig.~\ref{fig:prediction}c we show the probability for the three phases given by the QCNN along a particular cut in the phase diagram (the black, dashed path with arrows in Fig.~\ref{fig:prediction}b). 
% The crossings in the output probability align well with the location of the transition points. 

The cluster-Ising model has the special property that the phase diagram contains the fixed-point wavefunctions Eq.~\eqref{eq:FP}, when only one of $g_{zxz}, g_{zz},g_{x}$ is non-zero.  We remove this property by applying the trained QCNN to four additional time-reversal symmetric physical models that are previously unseen by the network.
To start, we consider a cluster model with a $Y$ field, namely $ H_1 = (1-\lambda)\sum_i Z_{i-1}X_{i}Z_{i+1} - \lambda\sum_iY_i$. The model has a transition from the SPT phase to the trivial phase at $\lambda = 1/2$, which is accurately captured by the QCNN as shown in Fig.~\ref{fig:compare_cases}a. Similarly, we consider $ H_2 = (1-\lambda)\sum_i Z_{i-1}Y_{i}Z_{i+1} - \lambda\sum_iY_i$ with a modified cluster coupling term. The transition at $\lambda = 1/2$ is also identified by the QCNN as shown in Fig.~\ref{fig:compare_cases}b. 

Next, we consider $ H_3= (1-\lambda)\sum_i X_{i-1}Y_{i}X_{i+1} + \lambda\sum_iY_i$. This Hamiltonian illustrates an intricate example where the correlation length diverges at $\lambda = 1/2$, but the system never leaves the trivial phase with respect to the $T$ symmetry. 
%This is because the diverging correlation length can be circumvented by adding some perturbation symmetric under $T$. 
As shown in Fig.~\ref{fig:compare_cases}c, the QCNN trained based on the representation $T$ correctly predicts the phase diagram. We emphasize that the system also has another time-reversal symmetry represented by $T'=(\prod_i Z_i) K$ which is responsible for a phase transition at $\lambda=1/2$: Under $T'$, the system belongs to distinct phases for $\lambda >1/2$ and $\lambda<1/2$.
%, giving rise to a phase transition at $\lambda = 1/2$. 
Such a transition can be captured if the QCNN is trained based on the representation $T'$. This can be easily verified by noting $H_2$, where the distinct phases are identified by the QCNN, is related to $H_3$ by a basis transformation.

The last example we consider is $H_4 = (1-\lambda)\sum_{i}H^{\text{bond}}_{2i} + \lambda\sum_{i}H^{\text{bond}}_{2i+1}$, describing an antiferromagnetic alternating-bond Heisenberg model. We denote the XXZ-coupling on each bond as $H^{\text{bond}}_{2i} = X_{2i}X_{2i+1}+Y_{2i}Y_{2i+1}+\Delta Z_{2i}Z_{2i+1}$. At the limit $\lambda = 0$ and $\lambda = 1$, the system is in two different dimerized states. Interestingly, the system has a time-reversal symmetry represented by an effective spin-1 $\pi$-rotation around the $y$-axis in the bulk followed by a complex conjugation~\cite{verresen:2017}, which protects the two dimerized limits as distinct phases. However, under the symmetry representation $T$, the two dimerized limits can be continuously connected without a phase transition and they both belong to the trivial phase~\cite{SM}. For sufficiently strong $\Delta$, the model exhibits an anitferromagnetic ordering at intermediate $\lambda$. In Fig.~\ref{fig:compare_cases}d we show the case of $\Delta=4$, exhibiting an intermediate SB phase in the vicinity of $\lambda = 1/2$. 
% Again, the QCNN captures the correct physics of the system. 
We note that the transition points predicted by the QCNN are slightly shifted away from the iDMRG phase boundary. This is reasonable given the relatively large correlation length near the phase boundary, which the 8-qubit QCNN cannot fully accommodate.

\begin{figure}
    \centering
    \includegraphics{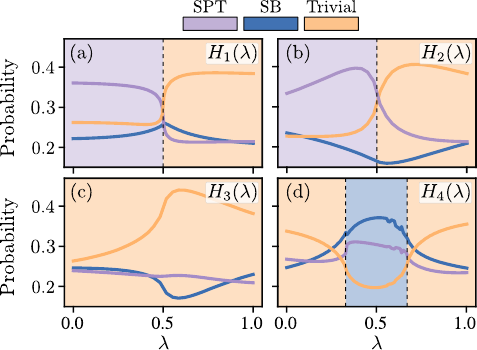}
    \caption{QCNN prediction on various microscopic models. In (a) and (b), a transition between an SPT and the trivial phase is detected at $\lambda = 1/2$. (c) The system is predicted to lie entirely in the trivial phase. A phase transition protected by a different symmetry occurs at $\lambda = 1/2$. (d) The bond-alternating Heisenberg model for $\Delta = 4$. The anitferromagnetic order is correctly detected by the QCNN. Details can be found in the text.}
    \label{fig:compare_cases}
\end{figure}
\textit{\textbf{Classification of other 1D symmetric phases.---}}The protocol we propose can be applied to generic 1D symmetric quantum phases (with additional symmetries, such as TI). The set of local symmetric unitary gates for a given symmetry representation forms a unitary Lie group and can be found by identifying all the symmetric generators of the group~\cite{SM}. 
For example, the cluster-Ising model is also protected by a $\mathbf{Z}_2\times\mathbf{Z}_2^T$ symmetry, i.e. the set of symmetries $\{I,\prod_i X_{i},K, \left(\prod_i X_{i}\right)K\}$, where $I$ is the identity. The symmetric two-qubit unitary is generated by $iZ_1Y_2$ and $iY_1Z_2$. Using the fixed-points Eq.~\eqref{eq:FP}, we can also train a QCNN that accurately characterizes the cluster-Ising model similar to Fig.~\ref{fig:prediction}a~\cite{SM}. This QCNN converges faster but is less powerful beyond the cluster-Ising model compared to the time-reversal case. Since $T\in \mathbf{Z}_2\times\mathbf{Z}_2^T$, the system has more disconnected phases due to a larger symmetry group. Therefore, more fixed-points are required to fully cover all the phases. However, as we see, imposing more symmetry can simplify the training process.

The non-existence of a QCNN without additional symmetries, such as TI, can also be proven for 1D systems with $ \mathbf{Z}_2\times\mathbf{Z}_2^T$ or $ \mathbf{Z}_2\times\mathbf{Z}_2$ symmetry~\cite{SM}. The requirement of additional symmetries is thus potentially applicable to the classification task of general symmetric phases.

\textit{\textbf{Discussion \& outlook.---}}
% The proposed protocol generates training data by randomly perturbing fixed-point wavefunctions with symmetric noise. 
The proposed method is reminiscent of data augmentation in classical machine learning for reducing over-fitting~\cite{Shorten:2019}. A key difference is that our training data set is entirely generated with perturbation. This is possible due to the notion of quantum phases. 
% Although we impose TI during the training of the QCNNs, we explicitly verify in the SM~\cite{SM} that the predictions of the trained QCNNs are robust under weak distorder, which can be attributed to a finite gap $\delta$ between the largest probability and the other probabilities in the output distribution of the QCNNs. 
Let $\delta $ be the finite gap between the largest probability and the other probabilities in the output distribution of a QCNN.
In practice, $\delta$ not only ensures that the QCNN's prediction is robust under weak perturbation, it also provides an estimate of the number of projective measurements required to accurately determine the prediction of the QCNN based on a majority vote. An error probability of $\epsilon<1$ can be achieved with more than $ 2\log\epsilon/\log(1-\delta^2)$ repetitions~\cite{SM}.

% We initialize the training with a particular choice of fixed-points, but we expect that sufficiently many random noise layers wash away the short-distance details. This also suggests that we can potentially use any state within the same phase to initialize the training, which can be helpful when some non-fixed-point states are more easily accessible in the experiments. Another 
The protocol can be further simplified by replacing the SB fixed-point, i.e. $\frac{1}{\sqrt{2}} (\ket{\cdots 000\cdots}+\ket{\cdots 111\cdots})$ with the asymmetric product state $\ket{\cdots 000\cdots}$ or $\ket{\cdots 111\cdots}$ which are easier to prepare. In the SM~\cite{SM}, we show that such replacement does not affect the performance of the trained QCNN on the time-reversal symmetric phases.

While physical observables that characterize 1D SPT phases are relatively well understood, probing the SPT order in higher dimensions is much more challenging~\cite{chen:2011a,zeletel:2014}. One exciting question is whether the proposed protocol can discover such an observable. 
% Due to the hardness of classically simulating high dimensional quantum systems, training on a quantum hardware might provide an advantage. 
Another interesting direction is to discover phase-classifying observables for intrinsic topological order, knowing that their fixed points can be efficiently prepared on quantum hardware~\cite{Satzinger2021,tant:2021,Bluvstein:2022,Liu:2022}. 
Although it has been shown that such an observable cannot exist in general~\cite{huang:2022}, it remains an open question whether imposing TI or other symmetries could help as for the 1D symmetric case discussed here. 

% Ideally, the symmetric noise layers should generate perturbations that explore the Hilbert space as uniformly as possible so that sampling biases can be avoided, which might be important for large-scale training. An open technical question is how to optimize the sampling strategies for general random symmetric unitaries. 

Under the current setup, the proposed method is unable to detect hidden phases that are not known a pirori. However, we observed that in some examples when a trained QCNN is implemented to classify an unknown phase, it gets confused by multiple phases with matching probability~\cite{SM}. It would therefore be intriguing to see whether this behavior is generic at large system size and whether such confusion could be used to identify existence of an unknown phase~\cite{nieu:2017}.
Another important question to study is the underlying principles for the phase detection behind a trained QCNN. Besides comparing it with some known analytical examples such as in Refs.~\cite{cong:2019,lake:2022}, a possible strategy would be to run the trained QCNN backward and use it as a generative model. Some properties of the trained QCNN may be inferred by examining the generated states. 
It will also be interesting to see whether the trainability of the classifiers can be improved by incorporating symmetry in the design of the classifiers~\cite{meyer:2022,Larocca:2022}.

%\begin{acknowledgments}

\textbf{\textit{Note added.---}}While finalizing our work, Ref.~\cite{lake:2022} appeared which also constructs QCNN phase classifiers on translationally invariant input states.

\textit{\textbf{Acknowledgement.---}}We thank Benoît Vermersch for helpful discussion. The DMRG simulations were performed using the TeNPy library~\cite{hauschild:2018}. We use the Jax library~\cite{jax2018github} for optimization.  Y.-J.L was supported by the Max Planck Gesellschaft (MPG) through the International Max Planck Research School for Quantum Science and Technology (IMPRS-QST). A.S. acknowledges support from a research fellowship from the The Royal Commission for the Exhibition of 1851.
Y.-J.L., M.K. and F.P. acknowledge support from the Deutsche Forschungsgemeinschaft (DFG, German Research Foundation) under Germany’s Excellence Strategy--EXC--2111--390814868 and DFG grants No. KN1254/1-2, KN1254/2-1, the European Research Council (ERC) under the European Union’s Horizon 2020 research and innovation programme (Grant Agreement No. 851161), as well as the Munich Quantum Valley, which is supported by the Bavarian state government with funds from the Hightech Agenda Bayern Plus.
{\par\textit{Data and materials availability:}} Data analysis and simulation codes are available on Zenodo upon reasonable request~\cite{zenodo}.

\bibliography{qcnn.bib}

\newpage
\widetext

%%%%%%%%%% Merge with supplemental materials %%%%%%%%%%
%%%%%%%%%% Prefix a "S" to all equations, figures, tables and reset the counter %%%%%%%%%%
\setcounter{equation}{0}
\setcounter{figure}{0}
\setcounter{table}{0}
\setcounter{page}{1}

\renewcommand{\theequation}{S\arabic{equation}}
\renewcommand{\thefigure}{S\arabic{figure}}
%\renewcommand{\bibnumfmt}[1]{[S#1]}
%\renewcommand{\citenumfont}[1]{S#1}
%%%%%%%%%% Prefix a "S" to all equations, figures, tables and reset the counter %%%%%%%%%%

\begin{center}
		{\fontsize{12}{12}\selectfont
			\textbf{Supplemental Material for  ``\papertitle''\\[5mm]}}
		
\end{center}
\normalsize\

The Supplemental Material is organized as follows: In Section~\ref{sm:sec:details}, we provide further details regarding the training and testing of the QCNN. In Section~\ref{sm:sec:nonsym}, we show further simulation results for the QCNN trained using a product state for the symmetry-breaking phase, instead of the symmetric state. In Section~\ref{sm:sec:othersym}, we show the results of applying the protocol to the system with $\mathbf{Z}_2\times\mathbf{Z}_2^T$ symmetry. In Section~\ref{sm:sec:TI}, we show the training results for uniform ansatz of the QCNN. In Section~\ref{sm:sec:extended}, we discuss the performence of the QCNN for different numbers of noise layers. In Section~\ref{sm:sec:nogo}, we prove non-existence results of phase classifying observables for time-reversal, $\mathbf{Z}_2\times\mathbf{Z}_2$ and $\mathbf{Z}_2\times\mathbf{Z}_2^T$ symmetry. In Section~\ref{sm:sec:so_interpretation}, we discuss how to correctly interpret the non-existence results in the context of string order parameters. In Section~\ref{sm:sec:avoidnogo}, we show how imposing TI allows us to avoid all the non-existence results in Section~\ref{sm:sec:nogo}. In Section~\ref{sm:sec:disorder}, we show the prediction accuracy of the QCNNs under weak disorder. In Section~\ref{sm:sec:heisenberg}, we discuss the dimerized Heisenberg chains. In Section~\ref{sm:sec:symmetrization}, we discuss the procedure of symmetrization to find symmetric generators for the local unitary. In Section~\ref{sm:sec:sample_estimate} we derive the upper bound for the error probability of a majority vote process.

\section{QCNN training and testing}\label{sm:sec:details}
In this section, we provide further details on the QCNN training and testing. In the main text, we focus on a QCNN that acts on $N$ qubits with $N = 4$ or 8 (see Fig.~1 in the main text). During the training, the pooling layers are absorbed into the two-qubit gates at the end of the the convolutional layers. A two-qubit gate is a $4\times 4$ unitary parametrized by 15 parameters as $\exp\left(-\frac{i}{2}\sum_{\rho,\gamma\in\{0,1,2,3\}}\theta_{\rho,\gamma}\hat O^{\rho}\otimes\hat O^{\gamma}\right)$, with the matrices $\hat O^0 = \mathbf{I},\ \hat O^1 = X,\ \hat O^2= Y$ and $\hat O^3= Z$. We set $\theta_{0,0} = 0$ to fix the phase degree of freedom of the gate.
\begin{figure}[h]
    \centering
    \includegraphics{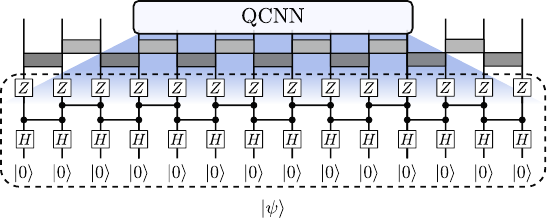}
    \caption{To simulate an $N$-site QCNN with $N = 8$ and $L_{\text{noise}} = 2$, we use a 14-qubit system with the QCNN acting on the middle 8 qubits. This avoids the finite-size effects. The absence of the finite-size effects can be checked for the input fixed-point function, e.g. $\ket{CS}$ with $Z_{i-1}X_iZ_{i+1} = -\ket{CS}$ in the bulk as shown. The $H$-gate is the Hadamard gate and the line connected by two dots is the controlled-$Z$ gate. The other fixed-point wavefunctions used in the work can also be verified to have no finite-size effects for the chosen QCNN.}
    \label{fig:sm:finite_size}
\end{figure}

We train the QCNN based on a log-softmax cross entropy loss function. Given a batch of $|B|$ input states and their labels $B = \{(\mathbf{p},l)\}$, the loss function for the batch is defined by
\begin{equation}
    \label{sm:eq:loss}
    L(B) = -\frac{1}{|B|}\sum_{(\mathbf{p},l)\in B}\log\bigg(\frac{e^{Cp_l}}{\sum\limits_j e^{Cp_j}}\bigg).
\end{equation}
In the above equation, $p_j$ is the $j^\mathrm{th}$ probability from the output bitstring distribution $\mathbf{p}$, $C$ is a constant used to set the desired scale of probability difference between different labels. In our experiments, we set $C=50$. Suppose the symmetric local unitary is generated by some Pauli strings $P_k$ such that the unitary is parametrized as $\exp(\sum_k \theta_kP_k)$ for some $\theta_k\in\mathbf{R}$. We sample the symmetric unitary by uniformly sampling $\theta_k\in(-\pi, \pi]$.

To implement the proposed protocol, we generate the training data for each phase by applying $L_{\text{noise}}=1$ or $L_{\text{noise}}=2$ layers of random symmetric two-qubit unitary to a fixed-point wavefunction of the phase. The training is done as follows: we start with $L_{\text{noise}}=1$ until we have reached 100\% test accuracy and we use the output QCNN to continue the training on data with $L_{\text{noise}}=2$. The optimization is performed using Adam optimizer~\cite{kingma2014adam}, with a learning rate of $5\times 10^{-4}$ for $L_{\text{noise}}=1$ and $1\times 10^{-4}$ for $L_{\text{noise}}=2$. To ensure convergence, for each training session of the 4-qubit (8-qubit) QCNN we generate 30000 (60000) samples for training and 1000 samples for testing. The batch size is 30 (50) and the number of epochs is chosen to be at most 12000. At the very end, we obtain a final test accuracy of the already-trained QCNN on 10000 engineered data with different $L_{\text{noise}}$. Note that we have not optimized the choice of the training sample size here. We expect much fewer training samples can be used to produce less optimal, yet reasonable results.

In the simulation, we simulate the application of an $N$-qubit QCNN to an infinite system by including $L_{\text{noise}}+1$ more qubits on the left and right of an $N$-qubit system, respectively. In total, the system contains $N+2(L_{\text{noise}}+1)$ qubits. An example of $N = 8$ and $L_{\text{noise}} = 2$ is depicted in Fig.~\ref{fig:sm:finite_size}, the system consists of 14 qubits in total. The QCNN only acts on the middle 8 qubits, such that any expectation values evaluated within these 8 qubits are the same as the expectation values evaluated in an infinite system. We can verify this with the circuit generating the cluster state as shown in Fig.~\ref{fig:sm:finite_size}.

We first train a 4-qubit QCNN using the protocol. We train with $L_{\text{noise}} = 1$ data and test on $L_{\text{noise}} = 1$ data as well. The trained QCNN yields a test accuracy of 87.21\%. It shows that the 4-qubit QCNN is starting to be able to recognize the phases (much better than a random guessing accuracy of $33\%$) but it is unable to do it accurately due to the small system size. 
We reproduce Fig.~2a and 2b in the main text using the 4-qubit QCNN in Fig.~\ref{fig:sm:qcnn_4q}. As we can see, near the fixed points where the correlation length of the system is small, the QCNN does a good job. The prediction becomes incorrect quickly when  approaching the phase boundaries. For an 8-qubit QCNN, we achieve a test accuracy 100\% for $L_{\text{noise}} = 1$ and 97.37\% for $L_{\text{noise}} = 2$. This shows the increased size of the QCNN allows it to distinguish states with a larger correlation length, as we expect.

Note that in the plots we show in the main text, we neglect the unsuccessful probability for the qubit label $11$. We can recover it by requiring that all the probabilities in the figures sum to 1. For example, we reproduce Fig. 2c in the main text with the unsuccessful probability included in Fig.~\ref{sm:fig:unsuccess}. 

\begin{figure}
    \centering
    \includegraphics{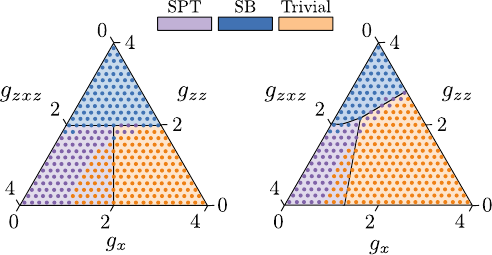}
    \caption{Phase diagrams predicted by a  4-qubit QCNN (c.f. Fig.2a and 2b in the main text for 8-qubit QCNN).}
    \label{fig:sm:qcnn_4q}
\end{figure}

\begin{figure}
    \centering
    \includegraphics{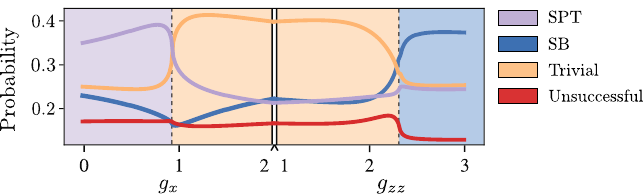}
    \caption{We reproduce the 8-qubt result shown in Fig.~2c in the main text here with the unsuccessful probability included. Now that all probabilities sum to 1 and we see that the unsuccessful probability remains small along this path.}
    \label{sm:fig:unsuccess}
\end{figure}

\section{QCNN prediction with modified training setup}\label{sm:sec:nonsym}
As we remark in the main text, the QCNN trained by the perturbed fixed-point wavefunction can pick up the long-range order of the system. We therefore expect that the fixed-point wavefunctions can be replaced by any wavefunctions that share the same long-range order of the phase. For example, we can simplify the training procedure by replacing the symmetric SB state $(\ket{\cdots 000\cdots}+\ket{\cdots 111\cdots})/\sqrt{2}$ by the asymmetric product state $\ket{\cdots 000\cdots}$ or $\ket{\cdots 111\cdots}$ which are much easier to prepare on a quantum hardware. 

We test this thought by performing the training of an 8-qubit QCNN for the time-reversal symmetry using the SB state $\ket{\cdots 000\cdots}$. Again, for each training session we use 60000 training samples and 1000 testing samples. We run the training with a batch size of 50 and 1000 epochs in total. At the end, we test the performance on 10000 data generated by the set of fixed-points involving the symmetric state, namely, the same test set for the QCNN in the previous section. For $L_{\text{noise}} = 1$, we achieve a test accuray of 100\% and for $L_{\text{noise}} = 2$ we achieve a test accuracy of 96.4\%. We see the performance is comparable to the QCNN trained with the symmetric fixed-point.

\section{Training with $\mathbf{Z}_2\times\mathbf{Z}_2^T$ symmetry}\label{sm:sec:othersym}
In this section, we apply the protocol to training the QCNN based on $\mathbf{Z}_2\times\mathbf{Z}_2^T$ generated by global spin flip and complex conjugation. We use the three fixed points provided in the main text. Note that, by using only three fixed points, we restrict to the phases that contain the fixed points, which only cover a subset of all the phases protected by the symmetry (with TI). Unlike the time-reversal case, we found the training in this case converges much quicker. We do not need to split the entire training into sessions with different $L_{\text{noise}}$. Instead, we directly train on the data with the prescribed $L_{\text{noise}}$. We used a training sample size of 30000 and a batch size of 30.

In this case, the symmetric local unitary is generated by the Pauli strings $iZ_1Y_2$ and $iY_1Z_2$. We first train a 4-qubit QCNN. The 4-qubit QCNN reaches a test accuracy of $99.98\%$ for $L_{\text{noise}}=1$. In Fig.~\ref{fig:sm:qcnn_4qPxT}, we show the phase diagram prediction similar to Fig.~2a and 2b using the 4-qubit QCNN. We see the QCNN does a nice job away from the phase boundary. Near the phase boundary, the QCNN again suffers from the large correlation length of the system and is less accurate. 

Next, we train an 8-qubit QCNN and the QCNN reaches a test accuracy of 100\% for both $L_{\text{noise}} = 1$ and 2. For $L_{\text{noise}} = 3$, the QCNN is not perfect and achieves 99.93\% test accuracy. The phase diagram prediction is shown in Fig.~\ref{fig:sm:qcnn_8qPxT}. Compared to the 4-qubit case, we see that the 8-qubit QCNN indeed improves significantly. Near the phase boundary, the 8-qubit QCNN is able to distinguish the phases accurately.
\begin{figure}
    \centering
    \includegraphics{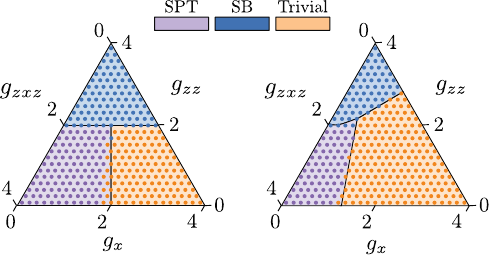}
    \caption{Comparing the theoretical phase diagram and the phase diagram predicted by the 4-qubit QCNN trained based on $\mathbf{Z}_2\times\mathbf{Z}_2^T$ symmetry.}
    \label{fig:sm:qcnn_4qPxT}
\end{figure}
\begin{figure}
    \centering
    \includegraphics{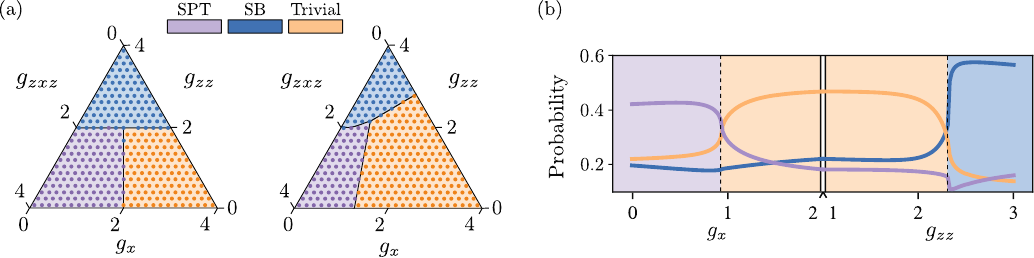}
    \caption{The 8-qubit QCNN trained based on $\mathbf{Z}_2\times\mathbf{Z}_2^T$ symmetry. Panel (a) and (b) are similar to Fig.~2 in the main text.}
    \label{fig:sm:qcnn_8qPxT}
\end{figure}

We summarize the test accuracy of the trained QCNN for different symmetries and number of qubits in Table.~\ref{tab:extended1} and~\ref{tab:extended2} with the values of $L_{\text{noise}}$, which we selected for the training of the QCNN in this section and Section~\ref{sm:sec:details}  marked in blue.

As we mentioned in the beginning of the section, the training is performed only with three fixed points under $\mathbf{Z}_2\times\mathbf{Z}_2^T$ symmetry. What happens if we apply the trained QCNN to predict an unknown phase? To experiment with this, we consider the following Hamiltonian with $\mathbf{Z}_2\times\mathbf{Z}_2^T$ symmetry
\begin{equation}
    H = (1-\lambda)\sum_i Z_{i-1}X_iZ_{i+1}-\lambda\sum_i Y_{i-1}X_iY_{i+1},
\end{equation}
where $\lambda\in[0,1]$. At $\lambda=0$, we recover the fixed point we used in the main text and the training. At $\lambda = 1$, the system in fact has a non-trivial SPT order where the symmetry of complex conjugation $K$ acquires a non-trivial fractionalization, but $K\prod_iX_i$ fractionalizes trivially. This phase is an unknown phase for the trained QCNN. Applying the trained QCNN to this model with different $\lambda$ yields the prediction as shown in Fig.~\ref{sm:fig:confusion}. We see that after a phase transition at $\lambda = 1/2$, the QCNN starts to get confused by multiple phases with a similar probability for $\lambda>1/2$. Whether this behavior is generic for other unknown phases using a larger QCNN is an interesting question to be investigated in the future.

\begin{figure}
    \centering
    \includegraphics[scale = 0.7]{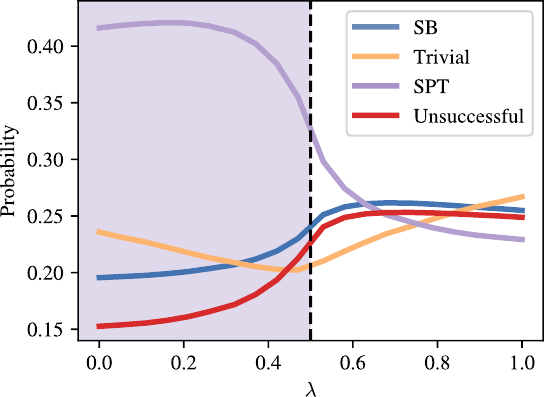}
    \caption{The 8-qubit QCNN trained with $\mathbf{Z}_2\times\mathbf{Z}_2^T$symmetric samples gets confused by multiple phases with matching probability when applied to an unknown phase.}
    \label{sm:fig:confusion}
\end{figure}

\begin{table}[]
\begin{tabular}{|ll|l|l|l|}
\hline
\multicolumn{2}{|l|}{$\textbf{Time-reversal}$}                                                  & $L_{\text{noise}}=1$ & $L_{\text{noise}}=2$ & $L_{\text{noise}}=3$ \\ \hline
\multicolumn{1}{|l|}{\multirow{3}{*}{4-qubit (90 parameters)}}         & \textcolor{blue}{$L_{\text{noise}} = 1$}   &   87.21\%                   & 68.69\% & 60.18\% \\ \cline{2-5} 
\multicolumn{1}{|l|}{}                                                 &  $L_{\text{noise}}=2$  &     84.66\%                 & 71.33\% & 64.44\% \\ \cline{2-5} 
\multicolumn{1}{|l|}{}                                                 &  $L_{\text{noise}}=3$ &     79.02\%                 & 69.44\% & 63.99\% \\ \hline
\multicolumn{1}{|l|}{\multirow{3}{*}{8-qubit (255 parameters)}}        & $L_{\text{noise}}=1$   &  100\%                    & 94.71\%  & 89.88\%  \\ \cline{2-5} 
\multicolumn{1}{|l|}{}                                                 & \textcolor{blue}{$L_{\text{noise}} = 2$} &   100\%                   & 97.37\% & 93.46\% \\ \cline{2-5} 
\multicolumn{1}{|l|}{}                                                 &  $L_{\text{noise}}=3$  &   100\%         & 97.09\%  & 93.62\% \\ \hline
\multicolumn{1}{|l|}{\multirow{3}{*}{Uniform 8-qubit (165 parameters)}} &  $L_{\text{noise}}=1$ &  100\%                    & 93.64\% & 87.61\%  \\ \cline{2-5} 
\multicolumn{1}{|l|}{}                                                 & \textcolor{blue}{$L_{\text{noise}} = 2$}   &   100\%                   & 96.76\%  & 92.46\%  \\ \cline{2-5} 
\multicolumn{1}{|l|}{}                                                 &  $L_{\text{noise}}=3$   &     100\%                 & 97.19\% & 93.37\% \\ \hline
\end{tabular}
\caption{The extended table for time-reversal symmetry. Different rows in the table correspond to different $L_{\text{noise}}$ used for the training. Different columns show the test accuracy on data with different $L_{\text{noise}}$ using the trained QCNN. The blue $L_{\text{noise}}$ are the ones we selected for the training in the main text and in Section~\ref{sm:sec:details} and~\ref{sm:sec:TI}. They are chosen based on the stopping criterion we propose.}
\label{tab:extended1}
\end{table}

\begin{table}[]
\begin{tabular}{|ll|l|l|l|}
\hline
\multicolumn{2}{|l|}{$\mathbf{Z}_2\times \mathbf{Z}_2^T$}                                                  & $L_{\text{noise}}=1$ & $L_{\text{noise}}=2$ & $L_{\text{noise}}=3$ \\ \hline
\multicolumn{1}{|l|}{\multirow{3}{*}{4-qubit (90 parameters)}}         & \textcolor{blue}{$L_{\text{noise}} = 1$}   &   99.98\%                   & 86.88\% & 83.17\% \\ \cline{2-5} 
\multicolumn{1}{|l|}{}                                                 &  $L_{\text{noise}}=2$  &     99.72\%                 & 98.08\% & 95.41\% \\ \cline{2-5} 
\multicolumn{1}{|l|}{}                                                 &  $L_{\text{noise}}=3$ &     99.62\%                 & 98.08\% & 95.94\% \\ \hline
\multicolumn{1}{|l|}{\multirow{3}{*}{8-qubit (255 parameters)}}        & $L_{\text{noise}}=1$   &   100\%                   & 97.67\%  &  94.41\% \\ \cline{2-5} 
\multicolumn{1}{|l|}{}                                                 & $L_{\text{noise}}=2$ &     100\%              & 100\% & 99.53\% \\ \cline{2-5} 
\multicolumn{1}{|l|}{}                                                 &  \textcolor{blue}{$L_{\text{noise}} = 3$}  &    100\%       & 100\%  & 99.93\%  \\ \hline
\multicolumn{1}{|l|}{\multirow{3}{*}{Uniform 8-qubit (165 parameters)}} &  $L_{\text{noise}}=1$ &     100\%                 & 91.97\% & 89.07\% \\ \cline{2-5} 
\multicolumn{1}{|l|}{}                                                 & $L_{\text{noise}}=2$   &     100\%                 & 100\% & 99.21\% \\ \cline{2-5} 
\multicolumn{1}{|l|}{}                                                 &  \textcolor{blue}{$L_{\text{noise}} = 3$}  &         100\%             & 100\% & 99.76\% \\ \hline
\end{tabular}
\caption{The extended table for $\mathbf{Z}_2\times \mathbf{Z}_2^T$ symmetry. The blue $L_{\text{noise}}$ are the ones we selected for the training in Section~\ref{sm:sec:othersym} and~\ref{sm:sec:TI}. They are chosen based on the stopping criterion we propose.}
\label{tab:extended2}
\end{table}
\section{Results for uniform QCNN}\label{sm:sec:TI}
In this section, we show results for training the uniform ansatz of the QCNN based on both the time-reversal symmetry and $\mathbf{Z}_2\times\mathbf{Z}_2^T$. The uniform ansatz has the same structure as the non-uniform case depicted in Fig.~1a in the main text, but with all unitaries at the same circuit layer being identical now. We compare the uniform ansatz to the non-uniform ansatz and observe that the performance is worse for the uniform ansatz when the depth of the convolutional layer is 3 as in Fig.~1a. To achieve a similar performance, we increase the depth of the convolutional layer from 3 to 5. For illustration, we will focus on the QCNN that acts on 8 qubits. The optimization is similar to the training of the time-reversal case. For each training session, we generate 30000 training samples and 1000 test samples. At the very end, we obtain a final test accuracy of the already-trained QCNN on 10000 engineered data with different $L_{\text{noise}}$.

For time-reversal symmetric systems, we obtain a QCNN that achieves 100\% on $L_{\text{noise}} = 1$ data and 96.76\% on  $L_{\text{noise}} = 2$ data when training with $L_\text{noise} = 2$. For the symmetry $\mathbf{Z}_2\times\mathbf{Z}_2^T$, we obtain a QCNN that achieves 100\% for both $L_{\text{noise}} = 1$ and 2 data. For $L_{\text{noise}} = 3$ it achieves an accuracy of 99.76\% when training with $L_\text{noise} = 3$. Both ansatzs are therefore comparable, see Table.~\ref{tab:extended1} and ~\ref{tab:extended2},

\section{Extended numerical results}\label{sm:sec:extended}

In the previous sections, we show the testing results of the trained QCNN for the time-reversal case. The QCNN is trained on a prescribed $L_{\text{noise}}$ picked by the stopping criterion we adopt. Namely, we increase $L_{\text{noise}}$ used in the training, until the test accuracy for the current $L_{\text{noise}}$ drops below certain threshold, which we choose to be 100\%. It is important that we do not over train the classifier with $L_{\text{noise}}$ that is too large, since this can potentially lead to over-fitting. The criterion aims to provide a stopping point where the classifier is reasonably converged and not over-fitted. We show the extended Table~\ref{tab:extended1} for time-reversal symmetry. In the extended table, we train and test the QCNN with $L_{\text{noise}} = 1,2$ and 3. The training is performed sequentially: we start the training with data generated by $L_{\text{noise}} = 1$ layer of noise. Once a convergence is reached we continue the training with $L_{\text{noise}} = 2$ and so on until we reach the prescribed $L_{\text{noise}}$. The color blue marks the $L_{\text{noise}}$ we use for the training based on the stopping criterion. We see that at the stopping points, the QCNNs have reasonably converged.

In contrast to the time-reversal case, we do not use the sequential training for obtaining the $\mathbf{Z}_2\times\mathbf{Z}_2^T$ results. Instead, we directly train on the data with a prescribed $L_{\text{noise}}$. We also picked a suitable $L_{\text{noise}}$ used in the training based on the same stopping criterion, which correspond to the $L_{\text{noise}}$ marked blue in the extended Table~\ref{tab:extended2}. For the 4-qubit case, we see that the performance of the QCNN can be further improved. This can be taken into account if we were to modify the threshold for the test accuracy from 100\% to 99\%.

\section{Proof of the non-existence results}\label{sm:sec:nogo}
In this section, we prove three non-existence results for classifying phases with physical observables when no additional symmetries are present. For systems with time-reversal symmetry only, we prove a non-existence result for a phase-classifying observable. We also prove that in general no QCNNs can classify systems with $\mathbf{Z}_2\times\mathbf{Z}_2^T$ or $\mathbf{Z}_2\times\mathbf{Z}_2$ symmetry, where $Z_2^T$ is an anti-unitary order-2 group.

\begin{prop}\label{prop1}
Let $U\subseteq \mathbf{C}_{N\times N}$ be a compact unitary group and $C_U$ be the centralizer of $U$ in $ \mathbf{C}_{N\times N}$. If $C_U \subseteq \{\lambda I|\lambda\in\mathbf{C}\}$, then for any $M\in\mathbf{C}_{N\times N}$
\begin{equation}
  \underset{\mu(u)}{\mathbf{E}} \ uMu^{\dag} = \frac{\Tr(M)}{N}I.
\end{equation}
where the average is taken over the Haar measure $\mu$ of $U$.
\end{prop}
\begin{proof}
We make use of the Haar integration. We have
\begin{equation}\label{eq:haar}
    A = \underset{\mu(u)}{\mathbf{E}} \ uMu^{\dag}  =\int uMu^{\dag} d\mu(u).
\end{equation}
Consider any $v\in U$, we have $vAv^{\dag} = A$ based on the invariance of Haar integration. So $A\in C_U$ is proportional to the identity. Taking the trace on Eq.~\eqref{eq:haar}, we get $\Tr(A) = \Tr(M)$. Knowing that the matrix is $N\times N$, we have $A = \frac{\Tr(M)}{N}I$.
\end{proof}
\begin{prop}\label{prop2}
Given two $n$-qubit states $\ket{\psi_a},\ket{\psi_b}$, there does not exist an operator $\hat O$ such that $\bra{\phi} \hat O\ket{\phi} > 0,\ \forall \ket{\phi}\in S_a$ and $\bra{\phi} \hat O\ket{\phi}\leq 0,\ \forall \ket{\phi}\in S_b$. The sets are defined as $S_a = \{u_1\otimes u_2\otimes\cdots \otimes u_m\ket{\psi_a}|u_i\in U\}$ and $S_b = \{u_1\otimes u_2\otimes\cdots \otimes u_m\ket{\psi_b}|u_i\in U\}$, where $U$ is a compact group of unitary operators that act on $k$ qubits and has the centralizer $C_U \subseteq \{\lambda I|\lambda\in\mathbf{C}\}$ and $mk = n$.
\end{prop}
\begin{proof}
The proof is adapted from Lemma~9 in Ref~\cite{huang:2022}, by combining it with the Proposition~\ref{prop1}. We elaborate the idea here. The result is established by contradiction. Suppose an operator $\hat O$ exists for $S_a$ and $S_b$ such that $\bra{\phi}\hat O\ket{\phi}>0,\ \forall \ket{\phi}\in S_a$ and $\bra{\phi}\hat O\ket{\phi}\leq 0,\ \forall \ket{\phi}\in S_b$. Next, we average over the Haar measure $\mu$ of $U$. From Proposition~\ref{prop1} this yields
\begin{equation}
  \underset{\mu(u)}{\mathbf{E}}   u_1^{\dag}\otimes u_2^{\dag}\otimes\cdots \otimes u_m^{\dag} \hat O u_1\otimes u_2\otimes\cdots \otimes u_m =\frac{\Tr\hat O}{2^n}I. 
\end{equation}
This follows by decomposing $\hat O$ into a linear combination of the basis operators, each of which is a tensor product of local operators supported on each qubit. Note that for an operator $\hat o$, we can always define a basis in the Hilbert space such that $\hat o = \sum_{ij}o_{ij}\ketbra{i}{j}$ and $o_{ij}$ is a matrix representation of $\hat o$, for which Proposition~\ref{prop1} applies. Since Haar integration is a linear map, Proposition~\ref{prop1} can be applied on each basis operator. Now we can define
\begin{align}
    o_a &\coloneqq \text{Inf}\{\bra{\phi}\hat O\ket{\phi} | \ket{\phi}\in S_a\},\\
    o_b &\coloneqq \text{Sup}\{\bra{\phi}\hat O\ket{\phi} | \ket{\phi}\in S_b\}.
\end{align}
Note that $U$ is compact and therefore closed. Hence, the infimum and supremum can be attained by some elements in $S_a$ and $S_b$, respectively. By definition, for $\ket{\psi_a}\in S_a$ and $\ket{\psi_b}\in S_b$ we have
\begin{align}
    \frac{\Tr\hat O}{2^n}=\underset{\mu(u)}{\mathbf{E}}  \bra{\psi_a} u_1^{\dag}\otimes u_2^{\dag}\otimes\cdots \otimes u_m^{\dag} \hat O u_1\otimes u_2\otimes\cdots \otimes u_m\ket{\psi_a}\geq o_a,\\
    \frac{\Tr\hat O}{2^n}=\underset{\mu(u)}{\mathbf{E}}  \bra{\psi_b} u_1^{\dag}\otimes u_2^{\dag}\otimes\cdots \otimes u_m^{\dag} \hat O u_1\otimes u_2\otimes\cdots \otimes u_m\ket{\psi_b}\leq o_b.
\end{align}
Now since $o_a>0$ and $o_b\leq 0$ by the assumption, we arrive at a contradiction that $\frac{\Tr\hat O}{2^n}\leq o_b < o_a \leq \frac{\Tr\hat O}{2^n}$.
\end{proof}

\begin{lemma}\label{sm:lemma1}
Consider a time-reversal symmetry represented by $\prod_i X_i K$ and two $n$-qubit states $\ket{\psi_a},\ket{\psi_b}$, there exists no  operator $\hat O$ such that $\bra{\phi} \hat O\ket{\phi} > 0\ \forall \ket{\phi}\in S_a$ and $\bra{\phi} \hat O\ket{\phi}\leq 0\ \forall \ket{\phi}\in S_b$, for the set $S_a = \{u_1\otimes u_2\otimes\cdots \otimes u_m\ket{\psi_a}|u_i\in Q\}$ and the set $S_b = \{u_1\otimes u_2\otimes\cdots \otimes u_m\ket{\psi_b}|u_i\in Q\}$, where $u_i$ acts on neighboring two qubits and $n = 2m$ and $Q$ is a symmetric unitary Lie group generated by $P=\{iZ_1, iZ_{2},iZ_1Y_{2},iY_1Z_{2}, iZ_1X_{2},iX_1Z_{2}  \}$.
\end{lemma}
\begin{proof}
Let $A\in \mathbf{C}_{4\times 4}$, note that $A$ commutes with all the elements of $Q$ if and only if $[A, p] = 0\ \forall p\in P$. Since $A$ can be decomposed into a linear combinations of $\sigma_1\sigma_2$, with $\sigma_1,\sigma_2\in\{I,X,Y,Z\}$. We first consider $p = iZ_1$ and $iZ_2$. Commuting with Pauli-$Z$ on two sites individually implies $A$ is a linear combination of $I, Z_1, Z_2, Z_1Z_2$. Since $A$ also commutes with $Y_1Z_2$ and $Z_1Y_2$, $A$ has to be proportional to the identity. Then by Proposition~\ref{prop2} we finish the proof. 
\end{proof}

If we consider any two states in the set of three fixed points considered in the main text, Lemma~\ref{sm:lemma1} implies there exists no QCNN that can be used to classify the time-reversal symmetric phases. 
We now proceed to prove non-existence results for phases protected by $\mathbf{Z}_2\times \mathbf{Z}_2^T$ and $\mathbf{Z}_2\times \mathbf{Z}_2$ using a similar idea. We first prove a non-existence result for distinguishing the SPT and the SB phases in systems protected by the symmetry $Z_2\times Z_2^T$ generated via a global spin flip and the complex conjugation. More precisely, we have the following
\begin{lemma}\label{sm:lemma2}
Let $\ket{\psi},\ket{\psi'}\in\{ \ket{\psi_{\text{SB}}},\ket{\psi_{\text{SPT}}}\}$ as defined in the main text and $\ket{\psi}\neq\ket{\psi'}$. Consider the set $S = \{u_1\otimes u_2\otimes\cdots \otimes u_m\ket{\psi}|u_i\in U\}$ and $S' = \{u_1\otimes u_2\otimes\cdots \otimes u_m\ket{\psi'}|u_i\in U\}$, where $u_i$ acts on neighboring three qubits and $n = 3m$. $U$ is a unitary Lie group generated by all the 3-qubit Pauli strings symmetric under $\mathbf{Z}_2\times \mathbf{Z}_2^T$ symmetry generated by the global spin flip and complex conjugation. There does not exist an operator $\hat O$ such that $\bra{\phi} \hat O\ket{\phi} > 0\ \forall \ket{\phi}\in S$ and $\bra{\phi} \hat O\ket{\phi} \leq  0\ \forall \ket{\phi}\in S'$. Furthermore, there does not exist a Hermitian operator $\hat D$ with $\Tr\hat D = 0$ and the number of supported qubits $n_{\hat D}<n$ such that $\bra{\phi} \hat D\ket{\phi} \neq 0\ \forall \ket{\phi}\in S$ or $\forall \ket{\phi}\in S'$.
\end{lemma}
\begin{proof}
We use a similar idea of the proofs from above. Again we will establish the proof by contradiction. To prove the first statement, we suppose such $\hat O$ exists.
\\
Let $U$ be the Lie group generated by all the anti-Hermitian Pauli matrices symmetric under $\mathbf{Z}_2\times \mathbf{Z}_2^T$, namely up to permutation we have $iXYZ, iYZX, \cdots$ and $iZYI, iZIY,\cdots$. Since $U$ is a compact Lie group, for any $8\times 8$ complex matrix $M$ we can define
\begin{equation}
     A  =\int u M u^{\dag} d\mu(u),
\end{equation}
where $\mu$ is the Haar measure of $U$. Consequently, $[A, u] = 0$ for $u\in U$. This is true if and only if $A$ commutes with all the generators of $U$, which implies
\begin{equation}\label{eq:average_z2z2T}
    A = c_0 I + c_1 X_1X_2X_3,
\end{equation}
for some $c_0,c_1\in\mathbf{C}$. Next, we consider the average
\begin{equation}
    \hat O' = \underset{\mu(u)}{\mathbf{E}} u_1\otimes u_2\otimes\cdots \otimes u_m\hat O u_1^{\dag}\otimes u_2^{\dag}\otimes\cdots \otimes u_m^{\dag}.
\end{equation}
$\hat O$ can be decomposed into a sum of at most $4^n$ Pauli strings. From Eq.~\eqref{eq:average_z2z2T} we deduce that $\hat O'$ is a linear combination of Pauli strings of $I$ or $X$. We note that any Pauli strings of $I$ and $X$ evaluate to the same value in $\ket{\psi_{\text{SPT}}}$ and $\ket{\psi_{\text{SB}}}$, so $\bra{\psi_{\text{SPT}}}\hat O'\ket{\psi_{\text{SPT}}}=\bra{\psi_{\text{SB}}}\hat O'\ket{\psi_{\text{SB}}}$. We can now apply exactly the same reasoning as in the proof of Proposition~\ref{prop2} to prove the first statement in the lemma.

Next  we assume there exists an operator $\hat D$ with $\Tr\hat D = 0$ and the number of supported qubits $n_{\hat D}<n$ such that $\bra{\phi} \hat D\ket{\phi} \neq 0\ \forall \ket{\phi}\in S$ or $\forall \ket{\phi}\in S'$.  Since $\Tr\hat D = 0$ implies its Haar average $\hat D'$ satisfies $\Tr\hat D' = 0$. $\hat D'$ is either 0 or consists of Pauli strings that have at least one $X$. Now knowing the support of $\hat D'$ is smaller than the total number of qubits in the system, we have $\bra{\psi_{\text{SPT}}} \hat D'\ket{\psi_{\text{SPT}}}=\bra{\psi_{\text{SB}}}\hat D'\ket{\psi_{\text{SB}}}=0$.
\\
Finally, we note that $U$ is connected and $\bra{\phi} \hat D\ket{\phi}$ is real-valued, it follows that there exists some $u_1\otimes u_2\otimes\cdots \otimes u_m$ that attains the mean value, resulting in a contradiction since there is some $\ket{\phi}\in S$ and $S'$ such that $\bra{\phi} \hat D\ket{\phi}=0$.
\end{proof}

The lemma suggests that there exists no QCNN we can use to distinguish or recognize the SPT and the SB phases protected by  $\mathbf{Z}_2\times \mathbf{Z}_2^T$. Note that the requirement $\Tr\hat D = 0$ is easy to meet. Since we can always remove the trace of an operator by redefining $\hat D - \frac{\Tr\hat D}{2^n}I$. The identity shift will not be relevant to characterize different phases if $\Tr\hat D/2^n\to 0$ happens faster than $\langle \hat D\rangle\to 0$ as $n_{\hat D}\to\infty$. For the QCNN described in the main text, we can distinguish the phases by choosing $\hat D$ to be the observable $\hat D = \hat U(\ketbra{s}{s}-\ketbra{s'}{s'})\hat U^{\dag}$ for bitstrings $s,s'$, where $\hat U$ is the QCNN circuit. We automatically have $\Tr\hat D = 0$.

Next, we proceed to the case of $\mathbf{Z}_2\times\mathbf{Z}_2$, generated by the spin flip on the even and the odd sites in the system. Note that the cluster state $\ket{\psi_{\text{SPT}}}$ we defined in the main text has an SPT order under the symmetry. We prove that a QCNN cannot recognize the SPT phase if no other symmetries are present.
\begin{lemma}\label{sm:lemma3}
Let $\ket{\psi_{\text{SPT}}}$ be the $n$-qubit cluster state defined in the main text. Consider the set $S = \{u_1\otimes u_2\otimes\cdots \otimes u_m\ket{\psi_{\text{SPT}}}|u_i\in U\}$, where $u_i$ acts on neighboring two qubits and $n = 2m$. The group $U = \{e^{i\alpha I + i\beta X_1 + i\gamma X_2 + i\delta X_1X_2}|\alpha,\beta, \gamma,\delta\in\mathbf{R}\}$ is symmetric under $Z_2\times Z_2$ symmetry generated by spin flip on even/odd sites. There does not exist a Hermitian operator $\hat O$ with $\Tr\hat O = 0$ and the number of supported qubits $n_{\hat O}<n/2$ such that $\bra{\phi} \hat O\ket{\phi} \neq 0\ \forall \ket{\phi}\in S$.
\end{lemma}
\begin{proof}
The proof is basically the same as Lemma~\ref{sm:lemma2}. We prove by contradiction. Suppose such $\hat O$ exists.
\\
Since $U$ is a compact Lie group, for any $4\times 4$ complex matrix $M$ we can define
$
     A  =\int u M u^{\dag} d\mu(u),
$
where $\mu$ is the Haar measure of $U$. Consequently, $[A, u] = 0$ for $u\in U$. This implies
\begin{equation}\label{eq:average_z2z2}
    A = c_0 I + c_1 X_1 + c_2 X_2 + c_3 X_1X_2,
\end{equation}
for some $c_0,c_1,c_2,c_3\in\mathbf{C}$. Next, we define
$
    \hat O' = \underset{\mu(u)}{\mathbf{E}} u_1\otimes u_2\otimes\cdots \otimes u_m\hat O u_1^{\dag}\otimes u_2^{\dag}\otimes\cdots \otimes u_m^{\dag}.
$
Again, we can decompose $\hat O$ in the basis of $4^{n_{\hat O}}$ Pauli strings. From Eq.~\eqref{eq:average_z2z2}, we deduce that $\hat O'$ is a linear combination of Pauli strings of $I$ or $X$. Since $\Tr\hat O = 0$ implies $\Tr\hat O' = 0$, $\hat O'$ contains Pauli strings that have at least one $X$. Now knowing the support of $\hat O'$ is smaller than half the total number of qubits in the system, the Pauli strings do not contain any symmetry of the system (i.e. $\mathbf{Z}_2\times\mathbf{Z}_2$). When they are applied to $\ket{\psi_{\text{SPT}}}$, the resulting state necessarily violates at least one of the cluster couplings so that $\hat O'\ket{\psi_{\text{SPT}}}$ is a linear combination of excited states for the cluster Hamiltonian and orthogonal to $\ket{\psi_{\text{SPT}}}$, we have $\bra{\psi_{\text{SPT}}} \hat O'\ket{\psi_{\text{SPT}}}=0$. Using the connectedness of $U$ and knowing $\bra{\phi} \hat O\ket{\phi}$ is real-valued, we know there exists some $u_1\otimes u_2\otimes\cdots \otimes u_m$ that attains the mean value, resulting in a contradiction, in that there is some $\ket{\phi}\in S$ such that $\bra{\phi} \hat O\ket{\phi}=0$.
\end{proof}

\section{Interpretation of the Lemmas for string order parameters}\label{sm:sec:so_interpretation}

\begin{figure}
    \centering
    \includegraphics[scale = 0.8]{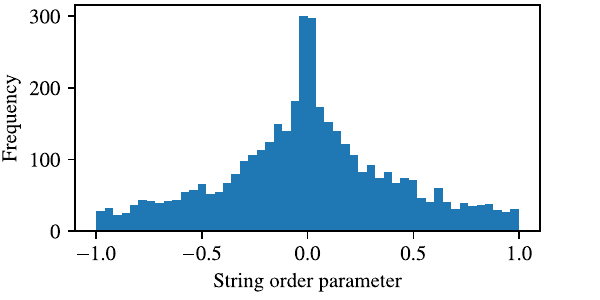}
    \caption{A histogram that shows the evaluation of a string order parameter $Z_1Y_2\left(\prod_{i=3}^{i=6} X_i\right)Y_7Z_8$ on a cluster state subject to two layers of $\mathbf{Z}_2\times\mathbf{Z}_2^T$ symmetric noise. A total of 4000 samples are taken.}
    \label{sm:fig:hist}
\end{figure}

In this section, we clarify how to interpret the lemmas we prove in the context of physical observables. The lemmas show that an observable that perfectly classifies the phases cannot exist in general without additional symmetries as such translational invariance. This forms a non-existence result for an observable that extracts the topological invariant of the phases (without utilizing multiple copies of the system). However, the statement can be overly restrictive in the context of string order parameters. 

The lemma says, that there does not exist a string order parameter that is zero, say in phase $A$ and strictly non-zero in a different phase $B$, i.e. there always exist some states in phase $B$ where the string order parameter is zero. The set of states in phase $B$ that have identically zero string order is measure-zero. This is because the vanishing of a string order parameter relies on a set of selection rules~\cite{Ppollmann:2012} when the system is in phase $A$. Without the selection rules a string order parameter is generically non-zero. To verify this, we consider a cluster state at the thermodynamic limit and apply two layers of $\mathbf{Z}_2\times\mathbf{Z}_2^T$ symmetric noise (each two-qubit gate is sampled as in one of the previous sections, but each two-qubit gate within a noise layer is now independently sampled). A string order parameter $Z_1Y_2\left(\prod_{i=3}^{i=6} X_i\right)Y_7Z_8$ of length 8 is then measured. Without any noise, the string order parameter attains +1 at the cluster state. With the noise, the string order parameter can attain any values between +1 and -1, as shown in Fig.~\ref{sm:fig:hist}. In particular, only for some instances the string order parameter is strictly zero in noisy cluster states. Moreover, due to the randomness of the noise, the average string order parameter is zero. However, the variance of it remains finite in the noisy cluster state.

This suggests that a string order parameter can still be used to classify phases in practice, up to a small fraction of states whose string order is vanishingly small (e.g. the ones near the center of Fig.~\ref{sm:fig:hist}). 

It is worth noting that, even though a string order parameter may exist and work well, we can neither directly use the string order to construct a cost function for optimization of a QCNN nor it is obvious how a QCNN that is linear in its input state and makes prediction based on a majority vote among the final measurement outcomes could reproduce the behavior of the string order parameter. An exception will be the existence of a string order that is zero in one phase and always non-negative (zero is only attained on a measure-zero set) in the other phase. However, this is excluded by the non-existence results we proved. When TI is enforced, it might be possible to find  such string order parameters for unitary on-site symmetries. In this case, QCNNs that learn them could exist. However, finding a string order parameter is not so simple for SPTs protected only by time-reversal symmetry. In the next section we will show that, for both unitary on-site symmetries and (antinunitary) time-reversal symmetry, there exists a more complicated non-local order parameter and it detects the topological invariant of the SPTs (e.g. it yields +1/-1 in different phases). Therefore, a QCNN can in principle be found.

\section{Circumventing the non-existence results}\label{sm:sec:avoidnogo}
\begin{figure}
    \centering
    \includegraphics{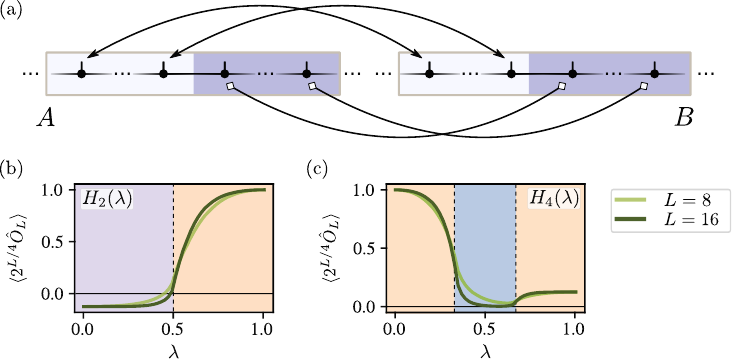}
    \caption{In (a) we illustrate the order parameter for the TI time-reversal systems. The solid circles represent the qubits on the chain. The double-arrow represents an operator $R = \frac{1}{2}\left(\ketbra{01}{01}+\ketbra{10}{10}+\ketbra{01}{10}+\ketbra{10}{01}\right)$. The double-square is a SWAP operator. The operator $R$ is applied between all the pairs across the light blue region of $A,B$, and the SWAP operator is applied between all the pairs across the light purple region of $A,B$. In (b) and (c), we verify the order parameter for TI time-reversal symmetric systems. We evaluate the order parameter on the ground states of $H_2$ and $H_4$ as defined in the main text. The order parameter is multiplied by $2^{L/4}$ for an easy comparison between cases with different $L$.}
    \label{fig:sm:TIop}
\end{figure}
In the main text, we claim that imposing a translational symmetry in additional to the time-reversal symmetry allows one to avoid the non-existence result by Lemma~\ref{sm:lemma1}. To see this, let us first consider the case between the trivial and the SPT phases. 
Suppose $A,B$ are two disjoint connected subregions in the unique ground state of a local and gapped 1D Hamiltonian. If $A,B$ are separated from each other further than the correlation length of the system, the reduced density matrix satisfies $\rho_{AB}\approx \rho_A\otimes \rho_B$. In TI system, we can choose the two separated regions to be identical copies. In Ref.~\cite{Ppollmann:2012}, an order parameter that dictates the topological invariants and requires two copies of the state is proposed for time-reversal systems: it attains different signs for the trivial and the SPT phases, and it vanishes in the SB phases. We can therefore use it for the phase classification task, as subsystems $A$ and $B$ serve as two copies of each other. A schematic diagram of this TI order parameter is shown in Fig.~\ref{fig:sm:TIop}a. In the case of the SB phases, we can always find a local field $\Delta_i$ such that $\bra{\psi}\Delta_i\ket{\psi}\neq 0$ ($\bra{\psi}\Delta_i\ket{\psi}= 0$) if $\ket{\psi}$ breaks (respects) the symmetry. For TI-SB systems, the phase can be probed by a non-zero value of $\Delta_i\Delta_{i+L}$, which vanishes in the trivial and the SPT phases for $L\to\infty$. The same observable can also be used to avoid the non-existence results in Lemma~\ref{sm:lemma2} for $\mathbf{Z}_2\times \mathbf{Z}_2^T$ with TI. While the observable constructed above has $\Tr\hat O\neq 0$ in general, by enlarging the SWAP operator sequence in Fig.~\ref{fig:sm:TIop}a, we can always define it in a way such that $\Tr\hat O/2^n\to 0$ happens exponentially faster than $\langle \hat O\rangle\to 0$ when it is evaluated on either SPT or the trivial phase as $n\to\infty$.

We numerically verify the order parameter for the TI time-reversal SPT in Fig.~\ref{fig:sm:TIop}b and \ref{fig:sm:TIop}c. The order parameter $\hat O_L$ of length $L$ consists of two operators proposed in Ref.~\cite{Ppollmann:2012} applied to two equally sized subsystems $A,B$ that are next to each other on the chain. We checked that taking $A,B$ apart from each other has little effects on the expectation value in the cases for which we will test it. We normalize the order parameter such that $||\hat O_L|| = 1$, where $||\cdot||$ is the operator norm. We evaluate the operator on $H_2$ and $H_4$ defined in the main text, with $L = 8$ and $16$. When $L$ is much larger than the correlation length of the system, $\hat O_L$ evaluates to 0 in the SB phase, and $\pm(\Tr\Lambda^4)^3/2^{L/4}$ in the trivial and the SPT phase, respectively. Here $\Lambda$ is the diagonal matrix whose diagonal entries are the Schmidt values obtained by cutting the chain into half. For cluster-state-like ground states, $\Lambda$ has two equal entries $1/\sqrt{2}$, we therefore expect $\langle 2^{L/4}\hat O_L\rangle = -1/8$. For a product state, $\Lambda$ only has one entry with value 1, we have  $\langle 2^{L/4}\hat O_L\rangle = 1$. When the entanglement-cut crosses a singlet, $\Lambda$ is doubly degenerate leading to $\langle 2^{L/4}\hat O_L\rangle = 1/8$. The plots give consistent results at these limits, as well as reasonable prediction of the critical points.

For the on-site symmetry $\mathbf{Z}_2\times \mathbf{Z}_2$ with TI, the trivial and SPT phases can be detected by a traceless observable proposed in Ref.~\cite{haegeman:2012}. The TI circumvents the non-existence result by Lemma~\ref{sm:lemma3}.

\section{Making predictions with weak disorder}
\label{sm:sec:disorder}
\begin{figure}
    \centering
    \includegraphics[scale = 0.9]{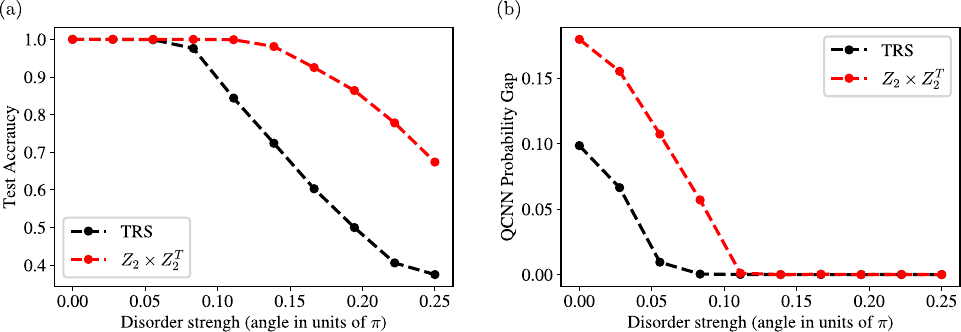}
    \caption{Predictions of the pre-trained 8-qubit QCNNs states generated with spatially disordered perturbation. (a) The test accuracy over 1000 test states generated by perturbing the fixed-points with one layer of independent two-qubit symmetric noise. Each noisy unitary is parametrized as described in Section~\ref{sm:sec:details}, with each angle parameter uniformly sampled from $[-\theta,\theta]$, where $\theta\geq 0$ characterizes the disorder strength (the $x$-axis). TRS stands for time-reversal symmetry. (b) Smallest difference between the probability of the predicted phase and the rest of the probabilities. A finite probability gap ensures some stability of the prediction against weak disorder.}
    \label{sm:fig:disorder}
\end{figure}
The non-existence results in Section~\ref{sm:sec:nogo} suggest that if the system is perturbed with noise that is not translationally invariant, the trained QCNNs in Section~\ref{sm:sec:details} and~\ref{sm:sec:othersym} will not be able to retain their high prediction accuracy even when the system only has a very short correlation length. However, when the breaking of the translation symmetry is weak, namely the perturbation is close to a symmetric operator, we still expect certain level of robustness. In particular, the robustness is characterized by a finite difference between the probability of the predicted phase and the other probabilities in the output distribution, i.e., by a finite probability gap. 

In Fig.~\ref{sm:fig:disorder}, we study this using the already-trained QCNNs from Section~\ref{sm:sec:details} and~\ref{sm:sec:othersym} to make prediction on the fixed-points perturbed by disordered noise. When the probability gap (Fig.~\ref{sm:fig:disorder}(b)) is finite we obtain a perfect test accuracy (Fig.~\ref{sm:fig:disorder}(a)). When the probability gap closes (Fig.~\ref{sm:fig:disorder}(b) gets close to zero), the prediction no longer guaranteed to be stable against disordered perturbation and the test accuracy in Fig.~\ref{sm:fig:disorder}(a) deviates from the perfect 100\%.

\section{Alternating-bond Heisenberg model}\label{sm:sec:heisenberg}
In this section, we show the two dimerzied limits in the alternating Heisenberg model belong to the same trivial phase under the time-reversal symmetry $T = \left(\prod_iX_i\right)K$. In the dimerzied limits, the ground state of the system is a product of local singlets. Since the wavefunctions are now real product states ($K$ acts trivially), we can examine how $T$ fractionalizes in the ground states by looking at how $\prod_iX_i$ fractionalizes. If we apply $\prod_i^L X_i$ to part of the system, we see that the operator acts as either $X$ or $I$ depending on whether the singlets are formed in the even or odd bonds. In both cases, the time-reversal SPT invariant is $X^2 = I^2 = +1$~\cite{schuch:2011,chen:2011}. So they belong the same trivial phase under $T$.

Next, we discuss how to connect the two dimerzied limits by a gapped symmetric Hamiltonian path.
Consider the Hamiltonian 
\begin{equation}\label{sm:eq:heisenberg}
H = (1-\lambda)\sum_i \mathbf{S}_{2i+1}\mathbf{S}_{2i+2}+\lambda \sum_i \mathbf{S}_{2i}\mathbf{S}_{2i+1},
\end{equation}
which corresponds to the case of $\Delta = 1$ in $H_4$ in the main text. The Hamiltonian has the same dimerized ground states at the limits $\lambda = 0,1$.
A transition exists at $\lambda = 1/2$. As we mention in the main text, this transition is protected by an $\pi$-rotation of the effective spin-1 (2-site unit cell) in the bulk followed by a complex conjugation. The transition can never be avoided if we keep this symmetry. 

\begin{figure}
    \centering
    \includegraphics[scale = 0.8]{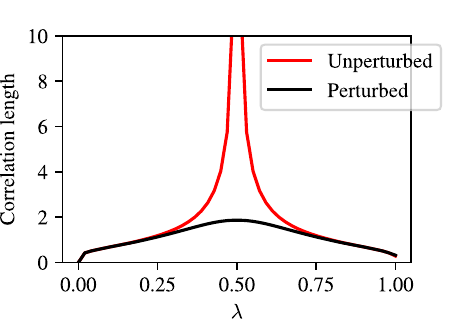}
    \caption{We verify the absence of transition in the bond-alternating Heisenberg model under time-reversal ($T$) symmetry, by examining the correlation length of the system. In the unperturbed case Eq.~\eqref{sm:eq:heisenberg}, the ground state has a diverging correlation length at $\lambda = 1/2$ (red solid line). This singular point can be avoided by adding a $T$-symmetric perturbation $\lambda(1-\lambda)\sum_i (-1)^iX_i$. The perturbed ground state continuously interpolates between the two dimerzied limits without a diverging correlation length (black solid line).}
    \label{fig:sm:heisenberg}
\end{figure}
However, we can avoid this gap-closing point by adding $T$-symmetric perturbation. A continuous path can be found by considering, e.g. $(1-\lambda)\sum_i \mathbf{S}_{2i+1}\mathbf{S}_{2i+2}+\lambda \sum_i \mathbf{S}_{2i}\mathbf{S}_{2i+1}+\lambda(1-\lambda) \sum_i (-1)^i h_i$, where $h_i = X_i$ respects the symmetry $T$. The added perturbation allows interpolation between the dimerzied ground states without gap-closing and therefore the two dimerized states can be continuously connected. The absence of a transition is verified  numerically by DMRG, as shown in Fig.~\ref{fig:sm:heisenberg}.

\section{Symmetrization of local quantum gates for generic symmetry groups}\label{sm:sec:symmetrization}
Suppose we are given an on-site symmetry of the form $U(g) = u(g)\otimes u(g)\otimes\cdots\otimes u(g)$ for $g\in G$, where $G$ is a finite group and $u(g)$ is a linear representation of $G$. An important question related to the proposed method is, how to find symmetric local quantum gates under a given (representation of the) symmetry group? A direct way for achieving this is by solving a set of linear equations (assuming the local gates have the same support as $u(g)$):
\begin{equation}
    [u(g),\sum_m c_m \hat o_m]=0,\ g\in G,\  c_m\in\mathbf{R}, \label{eq:find_sym}
\end{equation}
where $\hat o_m$ are the generators of the local unitary (usually anti-Hermitian Pauli strings). For time-reversal symmetry of the form $T = u(g)K$ with some $g\in G$ and $T^2 = 1$, we can enforce the symmetry locally and modify Eq.~\eqref{eq:find_sym} to be
\begin{equation}
    u(g)\left(\sum_m c_m \hat o_m\right)^*-\left(\sum_m c_m \hat o_m\right)u(g)=0,\  c_m\in\mathbf{R}, \label{eq:find_tsym}
\end{equation}

A convenient way to find the symmetric generators for on-site symmetry is by a twirling of the generators $\hat o_m$.
\begin{equation}
    \hat o'_m=\mathcal{T}[\hat o_m]_G = \frac{1}{|G|}\sum_{g\in G}u(g)\hat o_m u(g)^{\dag},
\end{equation}
which automatically ensures that $[\hat o'_m, u(g)] = 0$ for all $m$ and $g\in G$. This symmetrization procedure is used in Ref.~\cite{meyer:2022} to construct symmetric quantum circuit architectures.  For a time-reversal symmetry of the form $T = u(g)K$ with some $g\in G$ and $T^2 = 1$, we can define
\begin{equation}
    \hat o'_m = \hat o_m + u(g)(\hat o_m)^* u(g)^{\dag},
\end{equation}
which ensures $[\hat o'_m, u(g)K] = 0$ for all $m$.

The two procedures mentioned can be easily extended to symmetries described by compact Lie groups as well.

\section{The error probability for a majority vote}\label{sm:sec:sample_estimate}
In this section, we derive an upper bound for prediction made by a majority vote process. This provides an estimate for the sample size we need in order to make a reliable prediction in practice.

Suppose we have a discrete probability distribution with $M$ outcomes and the probabilities $(p_1,p_2,\cdots,p_M)$, such that $\sum_{i=1}^Mp_i = 1$ and  $p_i\geq 0$. The task is to find the label with the largest probability based on a majority vote among $2N+1$ independent samples from the distribution. We say a mistake is made if either (i) the majority vote cannot decide or (ii) the majority vote yields a wrong guess for the most probable label.

Without loss of generality, we sort the probability in descending order such that $p_1>p_2\geq p_3\geq \cdots\geq p_M$. The sampling process is multinomial and therefore we can write the error probability for the majority vote as
\begin{align}
    P &= \sum_{\substack{n_1+\cdots+n_M = 2N+1 \\ n_1\leq n_i,\ \text{for at least one } i\in[2, M]}}\binom{2N+1}{n_1,n_2,\cdots,n_M} p_1^{n_1}p_2^{n_2} \cdots p_M^{n_M}  
    \nonumber \\
    &\leq \sum_{\substack{n_1+\cdots+n_M = 2N+1 \\ n_1\leq n_i,\ \text{for at least one } i\in[2, M]}}\binom{2N+1}{n_1,n_2,\cdots,n_M} p_1^{n_1}p_2^{n_2+\cdots+n_M}  
    \nonumber \\
    &\leq \sum_{\substack{n_1+\cdots+n_M = 2N+1 \\ n_1\leq n_2+\cdots+n_M}}\binom{2N+1}{n_1,n_2,\cdots,n_M} p_1^{n_1}p_2^{n_2+\cdots+n_M}  \stackrel{k =
    n_2+\cdots+n_M}{=} \sum_{k = N+1}^{2N+1}\binom{2N+1}{k}p_1^{2N+1-k}p_2^k
    \nonumber \\
    &= \sum_{q = 0}^{N}\binom{2N+1}{N+1+q}p_1^{N-q}p_2^{N+1+q} 
    \nonumber \\
    &\leq p_1^{N}p_2^{N+1} \sum_{q = 0}^{N}\binom{2N+1}{N+1+q} = p_1^{N}p_2^{N+1} 2^{2N}\leq (4p_1p_2)^N.
\end{align}
In the third line above, we summed over the free multinomial indices to obtain a binomial coefficient. If we denote the probability gap by $\delta = p_1-p_2$, then we arrive at the bound
\begin{equation}
    P \leq (4p_1p_2)^N=\left((p_1+p_2)^2-\delta^2\right)^N\leq (1-\delta^2)^N.
\end{equation}
Suppose we want to ensure the error probability satisfies $P<\epsilon$ for some positive $\epsilon<1$, the sample size should at least be 
\begin{equation}
    2N+1\geq \frac{2\log\epsilon}{\log(1-\delta^2)}+1,
\end{equation}
as claimed in the main text.
\end{document}